\documentclass[a4paper]{article}

\usepackage{graphicx}
\usepackage{amsfonts}
\usepackage{amsthm}
\usepackage{enumerate}
\usepackage[scr]{rsfso}
\usepackage{units}
\usepackage{enumitem}
\usepackage{comment}
\usepackage{color}
\usepackage{textcomp}
\usepackage{hyperref}

\usepackage{tabularx,lipsum,environ,amsmath,amssymb}
\usepackage[boxed]{algorithm2e}
\usepackage{environ}
\usepackage{comment}
\usepackage{makecell}

\newif\ifarxiv

\newcommand{\arxivmath}{
	\ifarxiv
	$$
	\else
	$
	\fi}

\newcommand{\constSteklov}{St.\ Petersburg Department of Steklov Institute of Mathematics of the Russian Academy of Sciences, St.\ Petersburg, Russia}

\newcommand{\constGrant}{This research was supported by the Russian Science Foundation (project 16-11-10123)}

\newtheorem{brule}{Branching rule}
\newtheorem{rrule}{Reduction rule}
\newtheorem{hypothesis}{Hypothesis}
\newtheorem{claim}{Claim}

\makeatletter
\newcommand{\problemtitle}[1]{\gdef\@problemtitle{#1}}
\newcommand{\probleminput}[1]{\gdef\@probleminput{#1}}
\newcommand{\problemquestion}[1]{\gdef\@problemquestion{#1}}
\NewEnviron{problem}{
	\problemtitle{}\probleminput{}\problemquestion{}
	\BODY
	\par\addvspace{.5\baselineskip}
	\noindent
	{
		\begin{tabularx}{310pt}{@{\hspace{\parindent}} l X c}
		\multicolumn{2}{@{\hspace{\parindent}}l}{\@problemtitle} \\
		\textbf{Input:} & \@probleminput \\
		\textbf{Question:} & \@problemquestion
		\end{tabularx}}
	\par\addvspace{.5\baselineskip}
}

\DeclareMathOperator{\thr}{thr}
\renewcommand{\O}{\mathcal{O}}
\newcommand{\Ostar}[1]{\O^*\left(#1\right)}
 
\author{Ivan Bliznets\thanks{\constSteklov} \addtocounter{footnote}{-1} \\ \texttt{iabliznets@gmail.com}
	\and
		Danil Sagunov\footnotemark {} \\ \texttt{danilka.pro@gmail.com}}

\newtheorem{theorem}{Theorem}
\newtheorem{lemma}{Lemma}

\newtheorem{corollary}{Corollary}

\newtheorem{definition}{Definition}

\title{Solving Target Set Selection with Bounded Thresholds Faster than $2^n$\thanks{\constGrant}}

\arxivtrue

\begin{document}
\newpage
\maketitle

\begin{abstract}
In this paper we consider the \textsc{Target Set Selection} problem. The problem naturally arises in many fields like economy, sociology, medicine. In the \textsc{Target Set Selection} problem one is given a graph $G$ with a function $\thr: V(G) \to \mathbb{N} \cup \{0\}$ and integers $k, \ell$. The goal of the problem is to activate at most $k$ vertices initially so that at the end of the activation process there is at least $\ell$ activated vertices. The activation process occurs in the following way: (i) once activated, a vertex stays activated forever; (ii) vertex $v$ becomes activated if at least $\thr(v)$ of its neighbours are activated. The problem and its different special cases were extensively studied from approximation and parameterized points of view. For example, parameterizations by the following parameters were studied: treewidth, feedback vertex set, diameter, size of target set, vertex cover, cluster editing number and others.

Despite the extensive study of the problem it is still unknown whether the problem can be solved in $\Ostar{(2-\epsilon)^n}$ time for some $\epsilon >0$. We partially answer this question by presenting several faster-than-trivial algorithms that work in cases of constant thresholds, constant dual thresholds or when the threshold value of each vertex is bounded by one-third of its degree. Also, we show that the problem  parameterized by $\ell$ is W[1]-hard even when all thresholds are constant.

\end{abstract}
 
\section{Introduction}

In this paper we consider the \textsc{Target Set Selection} problem.
 In the problem one is given a graph $G$ with a function $\thr: V(G) \rightarrow \mathbb{N}\cup\{0\}$ (a \emph{threshold function}), and two integers $k, \ell$. The question of the problem is to find a vertex subset $S \subseteq V(G)$ (a \emph{target set}) such that $|S|\leq k$ and if we initially activate $S$ then eventually at least $\ell$ vertices of $G$ become activated. The activation process is defined by the following two rules: (i) if a vertex becomes activated it stays activated forever; (ii) vertex $v$ becomes activated if either it was activated initially or at some moment there is at least $\thr(v)$ activated vertices in the set of its neighbours $N(v)$. Often in the literature by \textsc{Target Set Selection} people refer to the special case of \textsc{Target Set Selection} where $\ell=|V(G)|$, i.e.\ where we need to activate all vertices of the graph. We refer to this special case as \textsc{Perfect Target Set Selection}.

\textsc{Target Set Selection} problem naturally arises in such areas as economy, sociology, medicine. Let us give an example of a scenario \cite{kempe2003maximizing,BenZwi2011} under which  \textsc{Target Set Selection} may arise in the marketing area. Often people start using some product when they find out that some number of their friends are already using it. Keeping this in mind, it is reasonable to start the following advertisement campaign of a product: give out the product for free to some people; these people start using the product, and then some friends of these people start using the product, then some friends of these friends and so on. For a given limited budget for the campaign we would like to give out the product in a way that eventually we get the most users of the product. Or we may be given the desired number of users of the product and we would like to find out what initial budget is sufficient. It is easy to see that this situation is finely modelled by the \textsc{Target Set Selection} problem.

The fact that \textsc{Target Set Selection} naturally arises in many different fields leads to a situation that the problem and its different special cases were studied under different names: \textsc{Irreversible $k$-Conversion Set} \cite{Centeno2011, Dreyer2009}, \textsc{P$_3$-Hull Number} \cite{Arajo2013}, \textsc{$r$-Neighbour Bootstrap Percolation} \cite{balogh2010bootstrap}, \textsc{$(k, \ell)$-Influence} \cite{Bazgan2014}, monotone dynamic monopolies \cite{peleg1998size}, a generalization of \textsc{Perfect Target Set Selection} on the case of oriented graphs is known as \textsc{Chain Reaction Closure} and \textsc{$t$-Threshold Starting Set} \cite{abrahamson1995fixed}. In \cite{Centeno2011}, Centeno et al.\ showed that \textsc{Perfect Target Set Selection} is NP-hard even when all threshold values are equal to two.

 There is an extensive list of results on \textsc{Target Set Selection} from  parameterized and approximation point of view. Many different parameterizations  were studied in the literature such as size of the target set, treewidth, feedback vertex set, diameter, vertex cover, cluster editing number and others (for more details, see Table~\ref{table:param_res}). Most of these studies consider the \textsc{Perfect Target Set Selection} problem, i.e.\ the case where $\ell=|V(G)|$. However, FPT membership results for parameters treewidth \cite{BenZwi2011} and cliquewidth \cite{hartmann2018target} were given for the general case of \textsc{Target Set Selection}. From approximation point of view, it is known that the minimization version (minimize the number of vertices in a target set for a fixed $\ell$) of the problem is very hard and cannot be approximated within  $\O(2^{\log^{1-\epsilon} n})$ factor for any $\epsilon > 0$, unless $\text{NP} \subseteq \text{DTIME}(n^{polylog(n)})$. This inappoximability result holds even for graphs of constant degree with all thresholds being at most two \cite{Chen2009}. Also, the maximization version of the problem (maximize the number of activated vertices for a fixed $k$) is NP-hard to approximate within a factor of $n^{1-\epsilon}$ for any $\epsilon > 0$ \cite{kempe2003maximizing}.

\begin{table}[!ht]
\begin{center}
\ifarxiv
\begin{tabular}{|>{\centering} m{2.55cm}|c|c|c|}
\else
\begin{tabular}{|c|c|c|c|}
\fi
\hline
Parameter & Thresholds & Result &  Reference\\
\hline

Bandwidth $b$ & general & $\Ostar{b^{\O(b\log b)}}$ & Chopin et al.~\cite{chopin2014constant} \\

Clique Cover Number $c$ & general & NP-hard for $c=2$ & Chopin et al.~\cite{chopin2014constant}\\

Cliquewidth $\text{cw}$ & constant & $\Ostar{(\text{cw} \cdot t)^{\O(\text{cw} \cdot t)}}$ & Hartmann~\cite{hartmann2018target}  \\

Cluster Editing Number $\zeta$ & general & $\Ostar{16^{\zeta}}$ & Nichterlein et al.~\cite{nichterlein2013tractable} \\

Diameter $d$ & general & NP-hard for $d=2$ & Nichterlein et al.~\cite{nichterlein2013tractable} \\

Feedback Edge Set Number $f$ & general & $\Ostar{4^f}$ & Nichterlein et al.~\cite{nichterlein2013tractable} \\

Feedback Vertex Set Number & general & W[1]-hard & Ben-Zwi et al.~\cite{BenZwi2011} \\

Neighborhood Diversity $\text{nd}$ & majority & $\Ostar{\text{nd}^{\O(\text{nd})}}$ & Dvo{\v{r}}{\'a}k et al.~\cite{dvovrak2016target} \\

& general & W[1]-hard & Dvo{\v{r}}{\'a}k et al.~\cite{dvovrak2016target} \\

Target Set Size $k$ & constant & W[P]-complete & \makecell{Abrahamson et al.~\cite{abrahamson1995fixed}, \\ Bazgan et al.~\cite{Bazgan2014}} \\

Treewidth $w$ & constant & $\Ostar{t^{\O(w\log w)}}$  & Ben-Zwi et al.~\cite{BenZwi2011}\\

& majority & W[1]-hard  & Chopin et al.~\cite{chopin2014constant}\\

Vertex Cover Number $\tau$ & general & $\Ostar{2^{(2^\tau+1)\cdot \tau}}$ & Nichterlein et al.~\cite{nichterlein2013tractable} \\

\hline
\end{tabular}
\end{center}
\caption{Some known results on different parameterizations of \textsc{Perfect Target Set Selection}. In the Thresholds column we indicate restrictions on the threshold function under which the results were obtained. Here $t$ denotes the maximum threshold value.}\label{table:param_res}
\end{table}

Taking into account many intractability results for the problem, it is natural to ask whether we can beat a trivial brute-force algorithm for this problem or its important subcase \textsc{Perfect Target Set Selection}. In other words, can we construct an algorithm with running time $\Ostar{(2-\epsilon)^n}$ for some $\epsilon > 0$. Surprisingly, the answer to this question is still unknown. Note that the questions whether we can beat brute-force naturally arise in computer science and have significant theoretic importance. Probably, the most important such question is SETH hypothesis which informally can be stated as:

\begin{hypothesis}[SETH]\label{hyp:seth}
 There is no algorithm for SAT with running time $\Ostar{(2-\epsilon)^n}$ for any $\epsilon>0$.
\end{hypothesis}

Another example of such question is the following hypothesis:
\begin{hypothesis}\cite{Pilipczuk2015}\label{hyp:hereditary}
 For every hereditary graph class $\Pi$ that can be recognized in polynomial time, the
\textsc{Maximum Induced $\Pi$-Subgraph} problem can be solved in $\Ostar{(2-\epsilon)^n}$ time for some $\epsilon>0$.
\end{hypothesis}

There is a significant number of papers~\cite{bliznets2013largest,pilipczuk2012finding,fomin2008solving,fomin2008minimum,cygan2014solving,razgon2007computing, fomin2011exact, cygan2014scheduling,cygan2010capacitated,binkele2011breaking,fomin2014enumerating} with the main motivation to present an algorithm faster than the trivial one. 

As in the stated hypotheses and mentioned papers, our goal is to come up with an algorithm that works faster than brute-force. We partially answer this question by presenting several $\Ostar{(2-\epsilon)^n}$ running time algorithms for \textsc{Target Set Selection} when thresholds, i.e.\ the values of $\thr(v)$, are bounded by some fixed constant and in case when the values of $\thr(v)-\deg(v)$, so-called \emph{dual thresholds}, are bounded by some fixed constant for every $v \in V(G)$. 
We think that this result may be interesting mainly because of the following two reasons. Firstly, the result is established for a well-studied problem with many applications and hence can reveal some important combinatorial or algorithmic structure of the problem. Secondly, maybe by resolving the asked question we could make progress in resolving hypotheses~\ref{hyp:seth}, \ref{hyp:hereditary}.

\textbf{Our results.}
In this paper, we establish the following algorithmic results.  

\textsc{Perfect Target Set Selection} can be solved in 
\begin{itemize}
    \item $\Ostar{1.90345^n}$ if for every $v \in V(G)$ we have $\thr(v) \le 2$;
    \item $\Ostar{1.98577^n}$ if for every $v \in V(G)$ we have $\thr(v) \le 3$;
    \item $\Ostar{(2-\epsilon_d)^n}$ randomized time if for every $v \in V(G)$ we have $\thr(v)\geq \deg(v)-d$.
\end{itemize}
\textsc{Target Set Selection} can be solved in 
\begin{itemize}
    \item $\Ostar{1.99001^n}$ if for every $v \in V(G)$ we have $\thr(v)\leq \lceil \frac{\deg(v)}{3} \rceil$;
    \item $\Ostar{(2-\epsilon_t)^n}$ if for every $v \in V(G)$ we have $\thr(v)\leq t$.
\end{itemize}
 
 We also prove the following lower bound.
 
  \textsc{Target Set Selection} parameterized by $\ell$ is W[1]-hard  even if
 \begin{itemize}
     \item $\thr(v)=2$ for every $v \in V(G)$.
 \end{itemize}

\section{Preliminaries}

\subsection{Notation and problem definition}

We use standard graph notation. We consider only simple graphs, i.e.\ undirected graphs without loops and multiple edges.  By $V(G)$ we denote the set of vertices of $G$ and by $E(G)$ we denote the set of its edges. We let $n=|V(G)|$. $N(v)$ denotes the set of neighbours of vertex $v \in V(G)$, and  $N[v]=N(v)\cup \{v\}$.  $\Delta(G)=\max_{v \in V(G)} \deg(v)$ denotes the maximum degree of $G$. By $G[F]$ we denote the subgraph of $G$ induced by a set $F$ of its vertices.  Define by $\deg_F(v)$ the degree of $v$ in the subgraph $G[F]$.

By $X_1 \sqcup X_2 \sqcup \ldots \sqcup X_m$ we denote the disjoint union of sets $X_1, X_2, \ldots, X_m$, i.e.\ $X_1\sqcup X_2\sqcup \ldots \sqcup X_m=X_1\cup X_2\cup \ldots \cup X_m$ with the additional restriction that $X_i \cap X_j=\emptyset$ for any distinct $i, j$.

For a graph $G$, threshold function $\thr$ and $X\subseteq V(G)$ we put $\mathcal{S}_0(X)=X$ and for every $i>0$ we define $\mathcal{S}_{i}(X)=\mathcal{S}_{i-1}(X) \cup \{v \in V(G) : |N(v) \cap \mathcal{S}_{i-1}(X)| \ge \thr(v)\}.$ We say that $v$ \emph{becomes activated in the $i^{th}$ round}, if $v \in \mathcal{S}_{i}(X) \setminus \mathcal{S}_{i-1}(X)$, i.e.\ $v$ is not activated in the $(i-1)^\text{th}$ round and is activated in the $i^\text{th}$ round. By \emph{activation process} yielded by $X$ we mean the sequence
$\mathcal{S}_0(X), \mathcal{S}_1(X), \ldots, \mathcal{S}_i(X), \ldots, \mathcal{S}_n(X).$ Note that $\mathcal{S}_n(X)=\mathcal{S}_{n+1}(X)$ as $\mathcal{S}_i(X) \subseteq \mathcal{S}_{i+1}(X)$ and  $n$ rounds is always enough for the activation process to converge. By $\mathcal{S}(X)$ we denote the set of vertices that eventually become activated, and we say that \emph{$X$ activates $\mathcal{S}(X)$ in $(G, \thr)$}. Thus, $\mathcal{S}(X)=\mathcal{S}_n(X)$.

We recall the definition of \textsc{Target Set Selection}.

\begin{problem}
    \problemtitle{{\scshape Target Set Selection}}
    \probleminput{A graph $G$ with thresholds $\thr: V(G) \to \mathbb{N} \cup \{0\}$, integers $k$, $\ell$.}
    \problemquestion{Is there a set $X\subseteq V(G)$ such that $|X|\le k$ and $|\mathcal{S}(X)| \ge \ell$?}
\end{problem}
We call a solution $X$ of \textsc{Target Set Selection} a \emph{target set of $(G, \thr)$}.

By \textsc{Perfect Target Set Selection} we understand a special case of \textsc{Target Set Selection} with $\ell=n$.  We call $X$ a \emph{perfect target set of $(G, \thr)$}, if it activates all vertices of $G$, i.e.\ $\mathcal{S}(X)=V(G)$.

Most of the algorithms described in this paper are recursive algorithms that use branching technique. Such algorithms are described by \textit{reduction rules}, that are used to simplify a problem instance, and \textit{branching rules}, that are used to solve an instance by recursively solving smaller instances. If a branching rule branches an instance of size $n$ into $r$ instances of size $n-t_1, n-t_2, \ldots, n-t_r$, we call $(t_1, t_2,\ldots, t_r)$ a \textit{branching vector} of this branching rule. By a \textit{branching factor} of a branching rule we understand a constant $c$ that is a solution of a linear reccurence corresponding to some branching vector of this rule; such constants are used to bound the running time of an algorithm following the rule with $c^n$.
Note that a branching rule may have multiple corresponding branching vectors and multiple corresponding branching factors.
By the \textit{worst branching factor} of a branching rule (or multiple branching rules, if they are applied within the same algorithm) we understand the largest among its branching factors.
We refer to \cite{Fomin:2010:EEA:1941886} for a more detailed explanation of these aspects.

In our work we also use the following folklore result.

\begin{lemma}\label{lemma:subsets}
For any positive integer $n$ and any $\alpha$ such that $0 < \alpha \le \frac{1}{2}$, we have $\sum\limits_{i=0}^{\lfloor \alpha n \rfloor} \binom{n}{i} \le 2^{H(\alpha)n},$
where $H(\alpha)=-\alpha \log_2(\alpha)-(1-\alpha)\log_2(1-\alpha)$.
\end{lemma}

\subsection{Minimal partial vertex covers}

\begin{definition}
Let $G$ be a graph. We call a subset $S \subseteq V(G)$ of its vertices a \textit{$T$-partial vertex cover of $G$} for some $T \subseteq E(G)$, if the set of edges covered by vertices in $S$ is exactly $T$, i.e.
$ T = \{ uv :\{u,v\} \cap S \neq \emptyset, uv \in E(G) \}.$

We call a $T$-partial vertex cover $S$ of $G$ a \textit{minimal partial vertex cover of $G$} if there is no $T$-partial vertex cover $S'$ of $G$ with $S' \subsetneq S$. Equivalently, there is no vertex $v \in S$ so that $S\setminus \{v\}$ is a $T$-partial vertex cover of $G$.
\end{definition}

The following theorem bounds the number of minimal partial vertex covers in graphs of bounded degree. We note that somewhat similar results were proven by Bj\"{o}rklund et al.~\cite{Bjrklund2012}.

\begin{theorem}\label{theorem:pvcf}
For any positive integer $t$, there is a constant $\omega_t < 1$ and an algorithm that, given an $n$-vertex graph $G$ with $\Delta(G) < t$ as input, outputs all minimal partial vertex covers of $G$ in $\Ostar{2^{\omega_tn}}$ time.
\end{theorem}

\begin{proof}
We present a recursive branching algorithm that lists all  minimal partial vertex covers of $G$. Pseudocode of the algorithm is presented in Figure~\ref{fig:familiess}. As input, the algorithm takes three sets $F,A,Z$ such that
 $F \sqcup A \sqcup Z = V(G)$. The purpose of the algorithm is to enumerate all minimal partial vertex covers that contain $A$ as a subset and do not intersect with $Z$. So the algorithm outputs all minimal partial vertex covers $S$ of $G$ satisfying $S \cap (A\sqcup Z)=A$. It easy to see that then $\texttt{minimal\_pvcs}(G,V(G),\emptyset,\emptyset)$ enumerates all minimal partial vertex covers of $G$.

\begin{figure}[!ht]
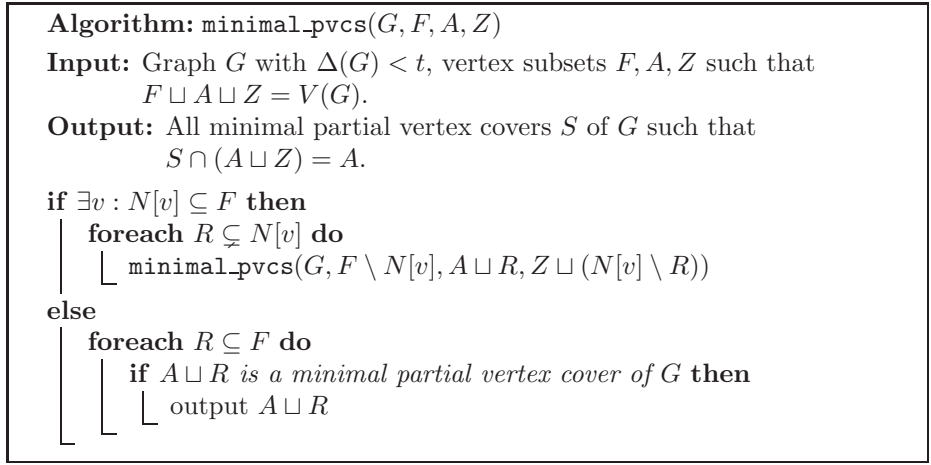

\centering
\begin{algorithm}[H]
 \TitleOfAlgo{$\texttt{minimal\_pvcs}(G, F, A, Z)$}
 \KwIn{Graph $G$ with $\Delta(G) < t$, vertex subsets $F,A,Z$ such that $F \sqcup A \sqcup Z = V(G)$.}
 \KwOut{All minimal partial vertex covers $S$ of $G$ such that $S \cap (A \sqcup Z)=A$.}
 \BlankLine
 \SetAlgoVlined
 \DontPrintSemicolon
 \eIf{$\exists v: N[v] \subseteq F$}{
    \ForEach{$R \subsetneq N[v]$}{
        $\texttt{minimal\_pvcs}(G, F \setminus N[v], A \sqcup R, Z \sqcup (N[v] \setminus R))$\;
    }
 }{
    \ForEach{$R \subseteq F$}{
        \If{$A \sqcup R$ is a minimal partial vertex cover of $G$}{
        output $A\sqcup R$\;
        }
    }
 }
 
\end{algorithm}
\caption{Algorithm enumerating all minimal partial vertex covers of a graph.}
\label{fig:familiess}
\end{figure}

 The algorithm uses only the following branching rule. If there is a vertex $v \in F$ such that $N(v) \subseteq F$ then  consider  $2^{|N[v]|}-1$ branches. In each branch, take some $R\subsetneq N[v]$ and run $\texttt{minimal\_pvcs}(G, F \setminus N[v], A \sqcup R, Z \sqcup (N[v] \setminus R))$. In other words, we branch on which vertices in $N[v]$ belong to minimal partial vertex cover and which do not. Note that if $S$ is a minimal partial vertex cover then it cannot contain $N[v]$, since otherwise $S \setminus \{v\}$ is its proper subset and covers the same edges. Hence, above branching consider all possible cases. Since $\Delta(G)<t$, the worst branching factor is $(2^{t}-1)^{\frac{1}{t}}$. 

If the branching rule cannot be applied then we apply brute-force on all possible variants of the intersection of the minimal partial vertex cover $S$ and the set $F$.  So we consider all $2^{|F|}$ variants of $S \cap F$, and filter out variants that do not correspond to a minimal partial vertex cover. Minimality of a partial vertex cover can be checked in polynomial time, so filtering out adds only a polynomial factor.

Note that we run brute-force only if every vertex in $F$ has at least one neighbour in $A \sqcup Z$, in other words, $A \sqcup Z$ is a dominating set of $G$. Since $\Delta(G)<t$, any dominating set of $G$ consists of at least $\frac{n}{t}$ vertices. Hence, $|F| \le \frac{(t-1)n}{t}$. This leads to the following upper bound on the running time of the algorithm:
\arxivmath \left(\left(2^{t}-1\right)^\frac{1}{t}\right)^{\frac{n}{t}}  
\cdot 2^\frac{(t-1)n}{t} \cdot n^{\O(1)}.
\arxivmath

Hence, we can put $\omega_t=\frac{1}{t^2}\log\left(2^t-1\right)+\frac{t-1}{t}<1.$
\end{proof}
 
\section{Algorithms for bounded thresholds}

\subsection{Algorithm for thresholds bounded by fixed constant}\label{sec:algo}

In this subsection we prove the following theorem.

\begin{theorem}
Let $t$ be a fixed constant. For \textsc{Target Set Selection} with all thresholds bounded by $t$ there is a $\Ostar{(2-\epsilon_t)^n}$-time algorithm, where $\epsilon_t$ is a positive constant that depends only on $t$.
\end{theorem}

Our algorithm consists of three main stages. In the first stage we apply some simple reduction and branching rules. If the instance becomes small enough we then apply brute-force and solve the problem. Otherwise, we move to the second stage of the algorithm. In the second stage we perform branching rules that help us describe the activation process. After that we move to the third stage in which we run special dynamic program that finally solves the problem for each branch. Let us start the description of the algorithm.

 \subsubsection{Stage I}
 
In the first stage our algorithm applies some branching rules. In each branch we maintain  the following partition of $V(G)$ into three parts $A, Z, F$. These parts have the following meaning: $A$ is the set of vertices that are known to be in our target set, $Z$ --- the set of vertices that are known to be not in the target set, $F$ --- the set of all other vertices (i.e.\ vertices about that we do not know any information so far). At the beginning, we have $A=Z=\emptyset$ and $F=V(G)$.

We start the first stage with exhaustive application of reduction rule \ref{rr1} and branching rule \ref{br1}.

\begin{rrule} \label{rr1}
If there is any vertex $v \in \mathcal{S}(A)$, but $v \notin A \sqcup Z$, then assign $v$ to $Z$.
\end{rrule}

Reduction rule~\ref{rr1} is correct as there is no need to put a vertex in a target set if it will become activated eventually by the influence of its neighbours.

\begin{brule} \label{br1} If  there is a vertex $v \in F$ such that  $\deg_F(v)\geq \thr(v)$ then arbitrarily  choose a subset $T \subseteq N(v) \cap F$ such that $|T|=\thr(v)$ and branch on the following branches:
\begin{enumerate}
    \item For each subset of vertices $S \subseteq T\cup \{v\}$ of size less than $\thr(v)$ consider a branch in which we put $S$ into $A$ and we put other vertices $T \cup \{v\} \setminus S$ into $Z$;
    \item Additionally consider the branch in which we assign all vertices from $T$ to $A$ and $v$ is assigned to $Z$.
\end{enumerate}
\end{brule}

It is enough to consider only above-mentioned branches. All other possible branches assign at least $\thr(v)$ vertices from $T \cup \{v\}$ to $A$, and we always can replace such branch with the branch assigning $T$ to $A$, since it leads to the activation of all vertices in $T\cup \{v\}$ and adds at most the same number of vertices into a target set. Branching rule~\ref{br1} considers $2^{\thr(v)+1}-\thr(v)-1$ options for $\thr(v)+1$ vertices, thus it gives the biggest branching factor of $(2^{t+1}-t-1)^{\frac{1}{t+1}}$ (here and below $t=\max_{v\in V(G)}\thr(v)$).

\begin{brule} \label{br2} If $|F| \leq \gamma n$, where $\gamma$ is a constant to be chosen later, then simply apply brute-force on how vertices in $F$ should be assigned to $A$ and $Z$. 
\end{brule}

If branching rule~\ref{br2}  is applied in all branches then the running time of the whole algorithm is at most $2^{\gamma n} (2^{t+1}-t-1)^{ \frac{(1-\gamma)n}{t+1} }$ and we do not need to use stages II and III, as the problem is already solved in this case.

\subsubsection{Stage II}

After exhaustive application of reduction rule~\ref{rr1} and branching rules~\ref{br1} and \ref{br2}, in each branch we either know the answer or we have the following properties:
\begin{enumerate}
\item $\Delta(G[F])<t$;
\item $|F|>\gamma n$;
\item $\mathcal{S}(A)\subseteq A\sqcup Z$.
\end{enumerate}

Now, in order to solve the problem it is left to identify the vertices of a target set that belong to $F$. It is too expensive to consider all $2^{|F|}$ subsets of $F$ as $F$ is too big. Instead of this direct approach (brute-force on all subsets of $F$) we consider several subbranches. In each such branch we almost completely describe the activation process of the graph. For each branch, knowing this information about the activation process, we find an appropriate target set by solving a special dynamic program in stage III.

Let $X$ be an answer (a target set). $X$ can be expressed as $X=A\sqcup B$ where $B \subseteq F$.
 At the beginning of the activation process only vertices in $\mathcal{S}_0(X)=X=A\sqcup B$ are activated, after the first round vertices in $\mathcal{S}_1(A\sqcup B)$ are activated, and so on. It is clear that $\mathcal{S}(A\sqcup B)=\mathcal{S}_n(A\sqcup B)$. Unfortunately, we cannot compute the sequence of $\mathcal{S}_i(A\sqcup B)$ as we do not know $B$. Instead we compute the sequence $P_0, P_1, \dots, P_n=P$ such that $P_i\setminus B=\mathcal{S}_i(X)\setminus B$ and $P_i\subseteq P_{i+1}$ for any $i$.

First of all, using Theorem~\ref{theorem:pvcf} we list all minimal partial vertex covers of the graph $G[F]$. For each minimal partial vertex cover $C$ we create a branch that indicates that $C\subseteq B$ and, moreover, $C$ covers exactly the same edges in $G[F]$ as $B$ does. In other words, any edge in $G[F]$ has at least one endpoint in $B$ if and only if it has at least one endpoint in $C$. Note that such $C$ exists for any $B$. One can obtain $C$ by removing vertices from $B$ one by one while it covers the same edges as $B$. When no vertex can be removed, then, by definition, the remaining vertices form a minimal partial vertex cover.

Put $P_0=A\sqcup C$. It is correct since $\mathcal{S}_0(X)\setminus B=A=P_0\setminus B$. We now show how to find $P_{i+1}$ having $P_i$. Recall that to do such transition from $\mathcal{S}_{i}(X)$ to $\mathcal{S}_{i+1}(X)$ it is enough to find vertices with the number of neighbours in $\mathcal{S}_i(X)$ being at least the threshold value of that vertex. As for $P_i$ and $P_{i+1}$, it is sufficient to check that the number of activated neighbours has reached the threshold only for vertices that are not in $B$. Thus any transition from $P_i$ to $P_{i+1}$ can be done by using a procedure that, given $P_i$ and any vertex $v \notin P_i$, checks whether $v$ becomes activated in the $(i+1)^{\text{th}}$ round or not, under the assumption that $v \notin B$.

Given $P_i$ it is not always possible to find a unique $P_{i+1}$ as we do not know $B$. That is why in such cases we create several subbranches that indicate potential values of $P_{i+1}$.  

Let us now show how to, for each vertex $v  \notin P_{i}$, figure out whether $v$ is in $P_{i+1}$ (see pseudocode in Figure~\ref{fig:is_activated_1}). Since we know $P_i$ and $P_i \subseteq P_{i+1}$, we assume that $v \notin P_i$.

If $|N(v) \cap P_i| \geq \thr(v)$ then we simply include $v$ in $P_{i+1}$. We claim that this check is enough for $v \in F$.

\begin{claim}
 If $v\in F\setminus B$, then $v$ becomes activated in the $i^\text{th}$ round if and only if $|N(v) \cap P_i| \geq \thr(v)$.
\end{claim}
\begin{proof}
We show that by proving that $\mathcal{S}_i(X)\cap N(v)=P_i\cap N(v)$ for every $v \in F\setminus B$. Note that $\mathcal{S}_i(X) \setminus B=P_i \setminus B$ by definition of $P_i$. So it is enough to prove that $\mathcal{S}_i(X)\cap N(v)\cap B= P_i \cap N(v)\cap B$, which is equivalent to $N(v) \cap B=P_i \cap N(v)\cap B$, as $B \subseteq \mathcal{S}_i(X)$. Since $v \notin B$, then any $uv \in E(G[F])$ is covered by $B$ if and only if $u \in B$. $C$ covers the same edges in $G[F]$ as $B$ does, and also $v\notin C$, hence $C \cap N(v)=B \cap N(v)$. Thus, since $C \subseteq P_0 \subseteq P_i$, we get $P_i \cap B \cap N(v)=P_i \cap C \cap N(v)=C \cap N(v) =B\cap N(v)$. 
\end{proof}

If $v\in B$, the decision for $v$ does not matter.
Thus if $v \in F$ and $|N(v) \cap P_i|<\thr(v)$, we may simply not include $v$ in $P_{i+1}$.

If $v \in Z$, at this point, we cannot compute the number of activated neighbours of $v$ exactly as we do not know what neighbours of $v$ are in $B$. Note that we do not need the exact number of such neighbours if we know that this value is at least $\thr(v)$. Thus we branch into $\thr(v)+1$ subbranches corresponding to the value of $\min\{|N(v)\cap B|, \thr(v)\}$, from now  on we denote this value as $dg(v)$.

On the other hand, we know all activated neighbours of $v$ that are in $V(G) \setminus F$ since $\mathcal{S}_i(X) \cap (V(G) \setminus F) = P_i \cap (V(G)\setminus F)$, as $B\subseteq F$. Let this number be $m=|N(v) \cap (P_i \setminus F)|$. So the number of activated neighbours of $v$ is at least $m+dg(v)$. Also there may be some activated neighbours of $v$ in $N(v)\cap P_i \cap F$. However, we cannot simply add $|N(v)\cap P_i \cap F|$ to $m+dg(v)$ since vertices in $P_i\cap B$ will be computed twice. So we are actually interested in the value of $|(N(v) \cap P_i \cap F) \setminus B|$. That is why for vertices from $N(v)\cap P_i \cap F$ we simply branch whether they are in $B$ or not. After that we compare $m+dg(v) + |(N(v)\cap P_i \cap F) \setminus B|$ with $\thr(v)$ and figure out whether $v$ becomes activated in the current round or not.

Note that once we branch on the value of $\min\{|N(v)\cap B|, \thr(v)\}$, or on whether $v \in B$ or not for some $v$, we will not branch on the same value or make a decision for the same vertex again as it makes no sense. Once fixed, the decision should not change along the whole branch and all of its subbranches, otherwise the information about $B$ would just become inconsistent.

Let us now bound the number of branches created. There are three types of branchings in the second stage:
\begin{enumerate}
\item Branching on the value of the minimal partial vertex cover $C$. By Theorem~\ref{theorem:pvcf}, there is at most $\Ostar{2^{\omega_t |F|}}$ such branches.
\item Branching on the value of $dg(v)=\min\{|N(v)\cap B|, \thr(v)\}$ with $v \in Z$. There is at most $(t+1)^{|Z |}$ such possibilities since  $t \geq \min\{|N(v)\cap B|, \thr(v)\} \geq 0$.
\item Branching on whether vertex $u$ is in $B$ or not.  We perform this branching only for vertices in the set $N(v)\cap P_i \cap F$ with $v \in Z$ only when its size is strictly smaller than $\thr(v) \le t$. Hence we perform a branching of this type on at most $(t-1)|Z|$ vertices.
\end{enumerate}

Hence, the total number of the branches created in stage II is at most \arxivmath 2^{\omega_t |F|} \cdot  (t+1)^{|Z|} \cdot 2^{(t-1)|Z|}\cdot n^{\O(1)}. \arxivmath

\begin{figure}[!ht]
\centering
\begin{algorithm}[H]
    \TitleOfAlgo{$\texttt{is\_activated}(G, \mathrm{thr}, A, Z, F, P_i, v)$}
    \KwIn{$G, \mathrm{thr}, A, Z, F$ as usual, $P_i$ such that $P_i \setminus B = \mathcal{S}_i(A\sqcup B) \setminus B$ for some $B$, and a vertex $v \notin P_i$.}
    \KwOut{True, if $v \notin B$ and $v \in \mathcal{S}_{i+1}(A \sqcup B)$; \newline False, if $v \notin B$ and $v \notin \mathcal{S}_{i+1}(A \sqcup B)$; \newline
    any answer, otherwise.}
    \BlankLine
    \SetAlgoVlined
    \DontPrintSemicolon
    \uIf{$|N(v) \cap P_i| \ge \thr(v)$}{
        \Return{True}\;
    }
    \ElseIf{$v \in F$}{
        \Return{False}\;
    }
    $m \longleftarrow |N(v) \cap (P_i \setminus F)|$\;
    branch on the value of $dg(v)=\min\left\{|N(v) \cap B|, \thr(v)\right\}$\;
    $m \longleftarrow m + dg(v)$\;
        \ForEach{$u \in P_i \cap N(v) \cap F$}{
            branch on whether $u \in B$\;
            \If{$u \notin B$}{
                $m \longleftarrow m + 1$\;
            }
        }
    \Return $m \ge \mathrm{thr}(v)$
\end{algorithm}
    \caption{Procedure determining whether a vertex becomes activated in the current round.}
    \label{fig:is_activated_1}
\end{figure}

\subsubsection{Stage III}
Now, for each branch our goal is to find the smallest set $X$ which activates at least $\ell$ vertices and agrees with all information obtained during branching in a particular branch. That is,
\begin{itemize}
\item $A\subseteq X, Z\cap X=\emptyset$ (branchings made in stage I);
\item $C\subseteq X$ (branching of the first type in stage II);
\item information about $\min\{|N(v)\cap B|,\thr(v)\}$ (second type branchings in stage II);
\item additional information whether certain vertices belong to $X$ or not (third type branchings in stage II).
\end{itemize}
  
  From now on we assume that we are considering some particular branching leaf. Let $A'$ be the set of vertices that are known to be in $X$ for a given branch and $Z'$ be the set of vertices known to be not in $X$ (note that $A\subseteq A'$ and $Z\subseteq Z'$). Let $Z=\{v_1,v_2,\dots, v_z\}$ and $F'=V(G)\setminus A' \setminus Z'=\{u_1,u_2,\dots, u_{f'}\}$. So actually  it is left  to find $B'\subseteq F'$ (in these new terms, $B=(A'\setminus A) \sqcup B'$) such that $|A'\sqcup B'|\leq k$, $|P\cup A'\cup B'| \geq \ell$ and for each $i \in \{1,2,\dots,z\}$ the value $\min\{\thr(v_i), |N(v_i)\cap B|\} $ equals $dg(v_i)$. This is true since the information obtained during branching completely determines the value of $P$.  

In order to solve the obtained problem we employ dynamic programming. We create a table $TS$ of size $ f' \times \ell \times (t+1)^z$. For all $B'_1$ such that $|(B'_1 \cup P) \cap \{u_1, u_2, \dots, u_i\}|=p$ and $\min\{\thr(v_j), |N(v_j)\cap ((A' \setminus A) \sqcup B'_1)|\}=d_j$, in the field $TS(i,p, d_1, d_2, \ldots, d_z)$  we store any set $B'_2$ of minimum size such that $A' \sqcup B'_1 \sqcup B'_2$ is a potential solution, i.e.\ $|\mathcal{S}(A' \sqcup B'_1 \sqcup B'_2)|=|(P\cup B'_1 \cup B'_2)|=|P\cap (V(G) \setminus F')|+p+|B'_2|\geq \ell$ and for every $j$ we have $\min\{\thr(v_j), |N(v_j)\cap ((A'\setminus A) \sqcup B'_1 \sqcup B'_2 )|\}= \min\{\thr(v_j), |N(v_j)\cap B'_2|+d_j\} = dg(v_j)$. Note that the choice of $B'_2$ depends only on values $i,p,d_1,d_2, \dots, d_z$, but not on the value of $B'_1$ directly. 
In other words, $TS(i,p, d_1, d_2, \ldots, d_z)$ stores one of optimal ways of how the remaining $f'-i$ vertices in $F'$ should be chosen into $B'$ if the first $i$ vertices in $F'$ was chosen correspondingly to the values of $p$ and $d_j$.

Note that for some fields in the $TS$ table there may be no appropriate value of $B'_2$ (there is no appropriate solution). In such cases, we put the corresponding element to be equal to $V(G)$. It is  a legitimate operation since we are solving a minimization problem. Note that the desired value of $B'$  will be stored as 
\arxivmath TS(0,0,\min\{|N(v_1) \cap (A' \setminus A)|, \thr(v_1)\},\ldots, \min\{|N(v_z) \cap (A' \setminus A)|, \thr(v_z)\}). \arxivmath

We assign $TS(f', p, dg(v_1), dg(v_2), \dots dg(v_z))=\emptyset$ for every $p$ such that $p+|P\cap (V(G) \setminus F')| \ge \ell$. We do this since values $p, dg(v_1), dg(v_2), \dots dg(v_z)$ indicate that $A' \sqcup B'_1$ is already a solution. In all other fields of  type $TS(f',\cdot,\cdots, \cdot )$ we put the value of $V(G)$.
We now show how to evaluate values $TS(i, p, d_1, d_2,\ldots, d_z)$ for any $i\geq 0$ smaller than $f'$. We can evaluate any $TS(i,\cdot, \cdot,\dots, \cdot )$  in polynomial time if  we have all values $TS(i+1,\cdot, \cdot,\dots, \cdot )$ evaluated.  For each $j \in \{1,2,\ldots, z\}$, let $d^{i+1}_j=\min\left\{\thr(v_j), d_j+|N(v_j) \cap \{u_{i+1}\}|\right\}$. In order to compute $TS(i, p, d_1, d_2,\ldots, d_z)$,  we need to decide whether $u_{i+1}$ is in a target set or not. If $u_{i+1}$ is taken into $B'$ then $d_j$ becomes equal to $d^{i+1}_j$ for each $j$, if it is not, none of $d_j$ should change. Hence,
$TS(i, p, \langle d_j \rangle)=\min\left[TS(i+1, p+1,\langle d^{i+1}_j \rangle) \cup \{u_{i+1}\}, \right.
 \left.TS(i+1, p+|P \cap \{u_{i+1}\}|, \langle d_j \rangle)\right].$

Since $0 \le d_j \le dg(v_j)$ for any $j$, the $TS$ table has $\Ostar{(t+1)^{|Z|}}$ fields. Each field of the table is evaluated in polynomial time. So the desired $B'$ is found (hence, the solution is found) in $\Ostar{(t+1)^{|Z|}}$ time for any branch fixed in stage II.
Stages II and III together run in $2^{\omega_t |F|} \cdot (t+1)^{|Z|} \cdot 2^{(t-1)|Z|}\cdot (t+1)^{|Z|}\cdot n^{\O(1)}$
time for any fixed subbranch of stage I.

Actually, the $(t+1)^{2|Z|}$ multiplier in the upper bound can be improved. Recall that it corresponds to the number of possible variants of $dg(v_j)$ and the number of possible variants of $d_j$. However, note that $d_j \leq dg(v_j)$. So after each of $dg(v_j)$ is fixed in stage II, for $d_j$ there is only $dg(v_j)+1$ options in stage III. Hence, each of the pairs $(d_j, dg(v_j))$ can be presented only in $\binom{t+2}{2}$ variants. This gives an improvement of the $(t+1)^{2|Z|}$ multiplier to a $\binom{t+2}{2}^{|Z|}$ multiplier. So, the upper bound on the running time in stages II and III becomes $\Ostar{2^{\omega_t |F|}\cdot\binom{t+2}{2}^{|Z|}\cdot 2^{(t-1)|Z|}}$.

We rewrite this upper bound in terms of $n$ and $|F|$. Since $|Z| \le n-|F|$, the upper bound is
\arxivmath 2^{\omega_t |F|}\cdot\binom{t+2}{2}^{n-|F|} \cdot 2^{(t-1)(n-|F|)}\cdot n^{\O(1)}. \arxivmath

Now we are ready to choose $\gamma$. We set the value  of $\gamma$ so that computation in each branch created at the end of stage I takes at most $\Ostar{2^{\gamma n}}$ time. Note that the upper bound on the running time required for stages II and III increases while the value of $|F|$ decreases. So we can find $\gamma$ as the solution of equation $2^{\gamma n}=2^{\omega_t \gamma n}\cdot\binom{t+2}{2}^{(1-\gamma) n}\cdot 2^{(t-1)(1-\gamma)n}$.  
Hence, $\gamma=\frac{(t-1)+\log_2\binom{t+2}{2}}{(t-\omega_t)+\log_2\binom{t+2}{2}} < 1,\text{ as } \omega_t < 1$. So the overall running time is \arxivmath 2^{\gamma n} (2^{t+1}-t-1)^{ \frac{(1-\gamma)n}{t+1} } \cdot n^{\O(1)}, \arxivmath which is $\Ostar{(2-\epsilon_t)^n}$ for some $\epsilon_t>0$ since $\gamma<1$. 
\subsection{Two algorithms for constant thresholds in the perfect case}

Here, we present two algorithms for special cases of \textsc{Perfect Target Set Selection} with thresholds being at most two or three. These algorithms use the idea that cannot be used in the general case of \textsc{Target Set Selection}, so the running times of these algorithms are significantly faster than the running time of the algorithm from the previous subsection. 
\begin{theorem}
\textsc{Perfect Target Set Selection} with thresholds being at most two can be solved in $\Ostar{1.90345^n}$ time.
\end{theorem} \begin{proof}
	To make the proof simpler, we firstly prove the following useful lemma.
	
	\begin{lemma}\label{lemma:bounded_to_equal}
		There is a polynomial-time algorithm that, given an integer constant $t\ge 2$ and a graph with thresholds $(G, \thr)$, where $\thr(v) \le t$ for every $v \in V(G)$, outputs a graph with thresholds $(G', \thr')$, such that $\thr' \equiv t$ and $|V(G')|=|V(G)|+t\cdot(t+1)$.
		Moreover, $(G, \thr)$ has a perfect target set of size $k$ if and only if $(G', \thr')$ has a perfect target set of size $k+t^2$. 
	\end{lemma}
	\begin{proof}
		We prove the lemma by providing a construction of graph $G'$.
		In this construction, $G'$ is obtained by introducing several new vertices and edges to $G$.
		It is as follows.
		
		For each integer $i \in [t]$, introduce a vertex $s_i$ and $t$ vertices $l_{i,1}, l_{i,2}, \ldots, l_{i,t}$ to $G$, then for each $j \in [t]$ introduce an edge between $s_i$ and $l_{i, j}$.
		Suchwise for each $i$ a star graph with $t$ leaves is introduced, with $s_i$ being the center vertex of the star graph.
		Finally, for each vertex of the initial graph $v \in V(G)$ introduce an edge between $v$ and $s_i$ for each $i \in [t-\thr(v)]$, that is, connect $v$ with the centers of the first $t-\thr(v)$ introduced star gadgets.
		Recall that in $G'$ we consider thresholds all-equal to $t \ge 2$.
		
		Let now show that if $(G', \thr')$ has a perfect target set of size $k'$, then $(G, \thr)$ has a perfect target set of size at most $k'-t^2$.
		Observe that any $l_{i,j}$ should be presented in every perfect target set of $G'$, since it has only one (that is, less than $t$) neighbour vertex. 
		Thus, any perfect target set of $G'$ contains all $t^2$ leaves of the star gadgets.
		For any fixed $i$, the initial activation of all $l_{i,j}$ activates $s_i$ in the first round.
		Hence, each star gadget provides a vertex that is always activated.
		Since each vertex $v \in V(G)$ is connected to exactly $t-\thr(v)$ vertices $s_i$ in $G'$, after activation of all $s_i$ $v$ requires $\thr(v)$ more vertices to become activated, that is the same as it does in $(G, \thr)$.
		Thus, if we remove the newly-introduced vertices from a perfect target set of $(G', \thr')$, we obtain a perfect target set of $(G, \thr)$.
		So if the perfect target set of $(G', \thr')$ has size $k'$, we obtain a perfect target set of $(G, \thr)$ of size at most $k'-t^2$.
		
		In the other direction, if $(G, \thr)$ has a perfect target set of size $k$, then $(G', \thr')$ has a perfect target set of size $k+t^2$. One can obtain a perfect target set of $(G', \thr')$ from a perfect target set of $(G, \thr)$ by adding all vertices $l_{i,j}$ to it.
	\end{proof}

	Lemma \ref{lemma:bounded_to_equal} allows us to reduce the case when all thresholds are bounded by two to the case when they are equal to two, introducing only a constant number of vertices.

	Let $(G, \thr)$ be a graph with thresholds, where all thresholds equal two. For this case, we present an algorithm with $\Ostar{1.90345^n}$ running time that finds a perfect target set of $(G, \thr)$ of minimum possible size.
	
	We set $\gamma=0.655984$. The algorithm consists of two parts. In the first part, the algorithm applies brute-force on all possible subsets $X \subseteq V(G)$ of size at most $(1-\gamma)n$, in ascending order of their size. If the algorithm finds $X$ that is a perfect target set, i.e.\ $\mathcal{S}(X)=V(G)$, then it outputs the set and stops. Otherwise, the algorithm runs its second part.

	The second part of the algorithm is a recursive branching algorithm that maintains sets $A,Z,F$ similarly to the algorithm in subsection \ref{sec:algo}. The branching algorithm consists of two reduction and two branching rules. Here, we reuse reduction rule \ref{rr1} and branching rule \ref{br1} from the previous subsection. Additionally, we introduce the following rules.
	
	\begin{rrule}\label{rule:small_deg_2}
		If there is a vertex $v \in F$ with $\deg_G(v) < 2$, assign $v$ to $A$.
	\end{rrule}
	
	Reduction rule \ref{rule:small_deg_2} is correct since such vertex cannot be activated other than being put in a target set. 
	
	\begin{brule}\label{rule:edge_two}
		If there are two vertices $u, v \in F$ with $uv \in E(G)$ and $\deg_G(u)=\deg_G(v)=2$, then consider three branches:
		\begin{itemize}
			\item $u \in Z$, $v \in A$;
			\item $u \in A$, $v \in Z$;
			\item $u, v \in A$.
		\end{itemize}
	\end{brule}
	
	Branching rule \ref{rule:edge_two} is correct since if none of $u,v$ is in a target set, none of them will eventually have two activated neighbours and thus the set cannot be completed to a perfect target set.
	
	If none of the rules can be applied, the algorithm applies brute-force on all $2^{|F|}$ possibilities of how vertices in $F$ should be assigned to $A$ and $Z$. This finishes the description of the second part and the whole algorithm. We now give a bound on its running time.
	
	By Lemma \ref{lemma:subsets}, the first part of the algorithm runs in $\Ostar{2^{H(1-\gamma)n}}=\Ostar{1.90345^n}$ time. If the algorithm does not stop in this part, then any perfect target set of $G$ consists of at least $(1-\gamma)n$ vertices and the second part is performed.
	
	Branching rules \ref{br1} and \ref{rule:edge_two} give branching vectors $(3,3,3,3,3)$ (five variants are considered for three vertices) and $(2,2,2)$ (three variants are considered for two vertices) respectively, and the second vector gives bigger branching factor equal to $\sqrt{3}$.
	
	Observe that if branching rules \ref{br1}, \ref{rule:edge_two} and reduction rules \ref{rr1}, \ref{rule:small_deg_2} cannot be applied, then $A \sqcup Z$ is in fact a perfect target set of $G$. Indeed, in that case $G[F]$ consists only of isolated vertices and isolated edges, as if there was a vertex $v\in F$ with $\deg_F(v) \ge 2$, then branching rule \ref{br1} would be applied. Note that if some vertex $v \in F$ is isolated in $G[F]$, then it has at least $\deg(v) \ge \thr(v)=2$ neighbours in $A \sqcup Z$, hence it becomes activated in the first round. Consider an isolated edge $uv \in G[F]$. Note that $u$ and $v$ cannot simultaneously have degree two in $G$, since branching rule \ref{rule:edge_two} excludes this case. It means that either $u$ or $v$ has degree at least three and thus has at least two neighbours in $A \sqcup Z$. Hence, it becomes activated in the first round. Since the other vertex has at least one neighbour in $A \sqcup Z$, at the end of the first round it will have at least two activated neighbours. Thus, it becomes activated no later than the second round. 
	
	We conclude that if we need to apply brute-force on $2^{|F|}$ variants, then $A \sqcup Z$ is a perfect target set of $G$. Hence, $|A\sqcup Z| \ge (1-\gamma)n$ and $|F| \le \gamma n$. It follows that the second part running time is at most $\sqrt{3}^{(1-\gamma)n} 2^{\gamma n}\cdot n^{\O(1)}=\Ostar{1.90345^n}$. So, the running time of the whole algorithm is $\max\{2^{H(1-\gamma)n}\cdot n^{\O(1)},  \sqrt{3}^{(1-\gamma)n} 2^{\gamma n}\cdot n^{\O(1)}\} = \Ostar{1.90345^n}$.
\end{proof}
 
\begin{theorem}
 \textsc{Perfect Target Set Selection} with thresholds being at most three can be solved in $\Ostar{1.98577^n}$ time.
\end{theorem} \begin{proof}

Here, we adapt the algorithm working for thresholds equal to two to the case when all thresholds equal three.
Again, we then use Lemma \ref{lemma:bounded_to_equal} to complete the proof.

Let $\gamma=0.839533$. 
At first, algorithm applies brute-force over all subsets of size at most $(1-\frac{2}{3}\gamma)n$ and stops if it finds a perfect target set among them. If the algorithm has not found a perfect target set on this step then we run a special branching algorithm.

As with thresholds equal to two we use branching rules \ref{br1}, \ref{rule:edge_two} and reduction rules \ref{rr1}, \ref{rule:small_deg_2}. The only difference is that now in reduction rule \ref{rule:small_deg_2}  and in branching rule \ref{rule:edge_two} we use constant $3$ instead of $2$. We also introduce a new branching rule for this algorithm.

\begin{brule}\label{rule:four_adj_two_three}
Let $v\in F$, $u, w \in N(v) \cap F$ and $\deg_G(v)=4, \deg_G(u)=\deg_G(w)=3$. Consider all branches that split $u,v,w$ between $A$ and $Z$ and assign at least one vertex to $A$.
\end{brule}
 The rule is correct as we omit only one branch that put all three vertices $u,v,w$ into $Z$. Note that if none of the vertices $u,v,w$ is activated initially then none of them will become activated. Hence, this branch cannot generate any perfect target set.

We apply the above-stated rules exhaustively. When none of the rules can be applied we simply apply brute-force on all possible subsets of $F$. That is the whole algorithm. Now, it is left to bound the running time of the algorithm.

The first part runs in $\Ostar{2^{H(1-\frac{2}{3}\gamma)n}}=\Ostar{1.98577^n}$ time. If the algorithm does not stop after the first part then any perfect target set of $G$ contains at least $(1-\frac{2}{3}\gamma)n$ vertices. Branching rules \ref{br1}, \ref{rule:edge_two}, \ref{rule:four_adj_two_three} give the following branching factors respectively: $12^{\frac{1}{4}}$ (since $12$ options are considered for $4$ vertices), $3^{\frac{1}{2}}$ ($3$ options for $2$ vertices) and $7^\frac{1}{3}$ ($7$ options for $3$ vertices). The biggest branching factor among them is $7^\frac{1}{3}$. 

Now, we bound the size of $F$ after exhaustive application of all rules.

\begin{lemma}
After exhaustive application of all rules $F$ consists of at most $\gamma n$ vertices.
\end{lemma}

\begin{proof}
Consider values of $A$, $Z$, $F$ when none of the rules can be applied. In this case we have that $\Delta(G[F]) < 3$. 

Note that our graph does not contain perfect target sets of size at most  $(1-\frac{2}{3}\gamma)n$. Otherwise algorithm would have finished working on the first step when it was brute-forcing over all subsets of size at most $(1-\frac{2}{3}\gamma)n$. Now, we start constructing a new perfect target set $P$ based on the structure of $A,F,Z$. Then, from the fact that $|P| > (1-\frac{2}{3}\gamma)n$, we obtain that $|F| \leq \gamma n$.

First of all, put $A \sqcup Z$ into a new perfect target set $P$. Let us show that degrees of vertices in the set $F'=F \setminus \mathcal{S}(A \sqcup Z)$ can only be three or four.  If $v \in F$ and $\deg_G(v)\geq 5$, then $v$ has at most two neighbours in $F$. Hence, it has at least three neighbours in $A \sqcup Z$ and so $v$ is in $\mathcal{S}(A\sqcup Z)$. 

Since $\Delta(G[F]) < 3$, it follows that $\Delta(G[F']) < 3$ also. Hence, any vertex $v \in F'$ with $\deg_G(v)=4$ requires one more activated neighbour to become activated. Also, $G[F']$ consists only of isolated paths and cycles. Consider any isolated path in $G[F']$. Observe
that any of its endpoints cannot have degree four in $G$, since otherwise it would have at least three neighbours in $\mathcal{S}(A \sqcup Z)$ and would be activated. Hence, all endpoints of all isolated paths are vertices
of degree three. Note that any endpoint has exactly two neighbours in $\mathcal{S}(A\sqcup Z)$. Since branching rule \ref{rule:small_deg_2} cannot be applied, any two endpoints
cannot be adjacent. Thus any isolated path in $G[F']$ consists of at least three vertices.

It means that the vertices that require two more activated neighbours to become activated are vertices of degree three that are not endpoints in any isolated path in $G[F']$. Note that if $u,v \in F'$ with $\deg_G(u)=\deg_G(v)=3$, and $u,v$ lie in the same isolated path or cycle $Q$ in $G[F']$, then there is at least two vertices of degree four in $Q$ between $u$ and $v$, since otherwise one of branching rules
\ref{rule:edge_two}   or  \ref{rule:four_adj_two_three} can be applied. Thus in any isolated path or cycle $Q$ in $G[F']$ the number of vertices that require at least two activated neighbours to become activated constitute at most one-third of the length of $Q$. We put all such vertices in the set $P$. There may be isolated paths or cycles left in $G[F']$ from which we have not put any vertex into $P$. For each such path or cycle we  choose an arbitrary vertex from it and put it into $P$. Note that from each isolated path or cycle in $G[F']$ we put no more than one-third of its vertices into $P$. Construction of $P$ is finished. 

From each isolated cycle or path we picked at least one vertex into $P$. The vertices that left require only one additional activated neighbour to become activated, in case of initially activated set $A\sqcup Z$. Hence, $P$ activates the whole graph. The size of $P$ is at most
$|A\sqcup Z|+\frac{1}{3}|F'|\le n-|F|+\frac{1}{3}|F| = n-\frac{2}{3}|F|$. It means that $n-\frac{2}{3}|F|\ge (1-\frac{2}{3}\gamma)n$. Hence, we proved $|F|\le \gamma n$.
\end{proof}

Using this lemma, we can bound the running time of the second part. The largest branching factor in the rules is $7^{\frac{1}{3}}$. Hence, the running time is at most   $7^{\frac{1}{3}(1-\gamma)n}2^{\gamma n}\cdot n^{\O(1)}=\Ostar{1.98577^n}$. Combining it with the running time of the first part we get that the overall running time is $\Ostar{1.98577^n}$.
\end{proof}

\subsection{Algorithm for thresholds bounded by one-third of degrees}

Here, we prove the following.

\begin{theorem} \label{thm:one_third}
	Let $G$ be a connected graph with at least three vertices. Assume that $\thr(v) \le \lceil \frac{\deg(v)}{3} \rceil$ for every $v \in V(G)$. Then there is a perfect target set of $(G, \thr)$ of size at most $0.45|V(G)|$.
\end{theorem} \begin{proof}
We prove this fact by induction on the number of vertices $n$ in $G$. 

If $G$ is connected and  $|V(G)|=3$ then any single vertex in $G$ forms a perfect target set. This is true since $\Delta(G) \le 2$ and thus the threshold value of any vertex of $G$ does not exceed $1$.

From now on $G$ is a connected graph on $n$ vertices with $n>3$. Let $n_{1}$ be the number of vertices in $G$ of degree one and $n_{\ge 2}$ be the number of vertices in $G$ of degree at least two, $n_1+n_{\ge 2}=n$. 

If $n_1 > n_{\ge 2}$, then there exist vertices $v,u_1,u_2 \in V(G)$ such that $vu_1, vu_2 \in E(G)$, $\deg(u_1)=\deg(u_2)=1$. Let $\rho(G, \thr)$ be the size of minimum perfect target set of $(G, \thr)$. Then $\rho(G, \thr) \le 1+\rho(G', \thr')$, where $G'=G \setminus v$ and $\thr'(u)=\thr(u)-|N(u) \cap \{v\}|$ for every $u \in V(G')$. Note that $\thr'(u) \le \lceil\frac{\deg_{G'}(u)}{3}\rceil$.

Let $G'$ consist of $k$ connected components  $C_1, C_2, \ldots, C_k$, where $k\ge 3$, since $C_1=\{u_1\}$, $C_2=\{u_2\}$. We assume that $|C_i| \le |C_{i+1}|$ for every $i \in \{1, 2, \ldots, k-1\}$.  We have that $\rho(G', \thr')=\sum\limits_{i=1}^k \rho(G'[C_i], \thr')$. Observe that if $|C_i| \le 2$, then $\rho(G'[C_i], \thr')=0$. Indeed, if $C_i=\{u\}$, then $\thr'(u)\le \deg_{G'}(u)=0$, and $u$ becomes activated in the first round. If $C_i=\{u, w\}$, then either $uv \in E(G)$ or $vw \in E(G)$, without loss of generality, say that $uv \in E(G)$. Also, $\deg_G(u), \deg_G(w) \leq 2$, thus $\thr(u)$ and $\thr(w)$ are not greater than one. Since $uv \in E(G)$, we have that $\thr'(u)=\thr(u)-1\le 0$. Thus $u$ becomes activated in the first round and as  $\thr'(w)\le \thr(w)\le 1$, then $w$ becomes activated no later than the second round. If $|C_i|\ge 3$, then, by induction, $\rho(G'[C_i], \thr')\le 0.45|C_i|$.

Hence, $\rho(G', \thr')\le \sum\limits_{i=m+1}^{k} \rho(G'[C_i], \thr')\le 0.45\sum\limits_{i=m+1}^{k} |C_i|$, where $m$ is such that $|C_m|\le 2$ and $|C_{m+1}|\geq 3$. Since $m\ge 2$, we have  $\rho(G', \thr')\le 0.45(|V(G')|-2)$. This implies that $\rho(G, \thr)\le 1+0.45(|V(G')|-2)=1+0.45(|V(G)|-1-2)<0.45|V(G)|$.

To handle the case $n_1 \le n_{\ge 2}$ (equivalent to $2n_1 \le n$) we use a combinatorial model proposed by Ackerman et al.\ in \cite{ackerman2010combinatorial}.
For each permutation $\sigma$ of vertices $V(G)$ we construct a perfect target set in the following way. We put vertex $v$ into the perfect target set if the number of neighbours to the left of $v$ in the permutation $\sigma$ is less than $\thr(v)$. It is easy to see that after such construction we get a perfect target set $P_{\sigma}$, as vertices will become activated from the left to the right. If we take a random permutation $\sigma$ among all permutations then the probability that a particular vertex $v$ ends up in $P_{\sigma}$ equals  $\frac{\thr(v)}{\deg(v)+1}$. Since $\thr(v)\le \lceil\frac{\deg(v)}{3}\rceil$, for a vertex of degree one the probability is bounded by $\frac{1}{2}$, for a vertex of degree two --- by $\frac{1}{3}$, for a vertex of degree three --- by $\frac{1}{4}$, for a vertex of degree four --- by $\frac{2}{5}$, etc. Observe that the highest probability bounds are for vertices of degree one and four, thus the expected value of the perfect target set size of $(G', \thr')$ is bounded by
$$\frac{1}{2}n_1+\frac{2}{5}n_{\ge 2} = \frac{1}{2}n_1 + \frac{2}{5}(n-n_1)= \frac{2}{5}n+\frac{1}{10}n_1 \le \frac{2}{5}n+\frac{1}{10}\cdot\frac{1}{2}n=\frac{9}{20}n.$$
Hence, there is at least one perfect target set of $(G, \thr)$ of size at most $0.45 n$.
\end{proof} 
\begin{corollary}
\textsc{Target Set Selection} with thresholds bounded by one-third of degree rounded up can be solved in $\Ostar{1.99001^n}$ time.
\end{corollary} 
\begin{proof}
Let $(G, \thr)$ and $k, \ell$ be an instance of \textsc{Target Set Selection} with $|V(G)|=n$ and $\thr(v)\le \lceil \frac{\deg(v)}{3}\rceil$ for every $v \in V(G)$. We are looking for $X \subseteq V(G)$ with $|X| \le k$ and $|\mathcal{S}(X)| \ge \ell$.

Consider subgraph $G'$ of $G$ consisting of all connected components of $G$ of size at least three. By Theorem \ref{thm:one_third}, $(G', \thr)$ has a perfect target set of size at most $0.45|V(G')| \le 0.45n$, hence it is enough to consider such $X$ that $|X \cap V(G')| \le 0.45n$. We apply brute-force on all such variants of $|X\cap V(G')|$. By Lemma \ref{lemma:subsets}, it takes $\Ostar{2^{H(0.45)n}}=\Ostar{1.99001^n}$ time.

When $|X\cap V(G')|$ is fixed, it is left to consider connected components of $G$ of size less than three. Note that if we already have $|X \cap V(G')| \le k$ and $|\mathcal{S}(X\cap V(G'))| \ge \ell$, we may set $X=X \cap V(G')$ and stop. Otherwise, we should consider adding vertices from connected components of size one or two to $X$. Adding a vertex from a connected component of size one, i.e.\ isolated vertex, increases the number of activated vertices by one, and adding a vertex from a component of size two increases this number by two. Thus we greedily assign a single vertex from each component of size two to $X$, but no more than $k-|X\cap V(G')|$ in total. If after that the size of $X$ is still less than $k$, we assign as many isolated vertices of $G$ to $X$ as we can. Then we finally check whether $|\mathcal{S}(X)| \ge \ell$.

The greedy part of the algorithm runs in polynomial time for each variant of $|X\cap V(G')|$. Hence, the whole algorithm runs in $\Ostar{1.99001^n}$ time.
\end{proof}

\section{Algorithm for bounded dual thresholds}

Let $(G, \thr)$ be a graph with thresholds. By \textit{dual threshold} of vertex $v \in V(G)$ we understand the value $\overline{\thr}(v)=\deg(v)-\thr(v)$. In terms of dual thresholds, $v$ becomes activated if it has at most $\overline{\thr}(v)$ not activated neighbours.
For bounded dual thresholds we prove the following theorem.

 \begin{theorem}
 For any non-negative integer $d$, \textsc{Perfect Target Set Selection} with dual thresholds bounded by $d$ can be solved in $\Ostar{(2-\epsilon_d)^n}$ randomized time for some $\epsilon_d > 0$.
 \end{theorem}
   \begin{proof}

 In terms of dual thresholds, we can consider the activation process as a vertex deletion process, where activated vertices are deleted from the graph. With this consideration, activation process goes in the  following way. Firstly, the target set is deleted from the graph. Then, in each consecutive round, a vertex $v$ is deleted from the remaining graph if it has at most $\overline{\thr}(v)$ neighbours remaining. When the process converges, vertices in the remaining graph are the vertices that are not activated. Thus the target set is perfect if and only if the remaining graph is empty.
 
 If $\overline{\thr}(v)=d$ for each $v\in V(G)$. Then, a vertex is deleted from the remaining graph if it has at most $d$ neighbours remaining. By definition of $d$-degeneracy, a graph becomes empty after such process if and only if it is $d$-degenerate. Thus, a target set $X$ is perfect if and only if $G\setminus X$ is $d$-degenerate. Hence, if all dual thresholds are equal to $d$, finding a maximum $d$-degenerate induced subgraph of $G$ is equivalent to finding a minimum perfect target set of $G$.
 
 In \cite{pilipczuk2012finding}, Pilipczuk and Pilipczuk presented an algorithm that solves \textsc{Maximum Induced $d$-Degenerate Subgraph} problem in randomized $(2-\epsilon_d)^n \cdot n^{\O(1)}$ time for some $\epsilon_d>0$ for any fixed $d$. Hence, instances of \textsc{Perfect Target Set Selection} where all dual thresholds are equal to $d$ can be solved in the same running time. Furthermore, one can straight-forwardly show that this algorithm can be adjusted to work when all dual thresholds are not necessarily equal, but do not exceed $d$.
 \end{proof} 
 \section{Lower bounds}

 \subsection{ETH lower bound}
First of all, we show a $2^{o(n+m)}$ lower bound
 for \textsc{Perfect Target Set Selection}, where $m$ denotes the number of edges in the input graph. We have not found any source that claims this result. Thus, for completeness, we state it here. The result follows from the reduction given by Centeno et al.\ in \cite{Centeno2011}. They showed a linear reduction from a special case of \textsc{3-SAT}, where each variable appears at most three times, to \textsc{Perfect Target Set Selection} where thresholds are equal to two and maximum degree of the graph is constant. Note that in their work they refer to the problem as \textsc{IRR$_2$-Conversion Set}.
 
\begin{theorem}
\textsc{Perfect Target Set Selection} cannot be solved in $2^{o(n+m)}$ time unless ETH fails, even when thresolds are equal to two and maximum degree of the graph is constant.
\end{theorem} \begin{proof}

\textsc{3-Bounded-3-SAT} is a version of \textsc{3-SAT} with a restriction that each variable appears at most three times in a formula. It is a well-known fact that an instance of \textsc{3-SAT} with $n$ variables and $m$ clauses can be transformed into an instance of \textsc{3-Bounded-3-SAT} with $\mathcal O(m)$ variables and $\mathcal O(m)$ clauses, in polynomial time. Then, according to the Exponential-Time Hypothesis with Sparsification Lemma, it follows that \textsc{3-Bounded-3-SAT} cannot be solved in $2^{o(n+m)}$ time.

In Theorem 2 in \cite{Centeno2011} Centeno et al.\ have shown how to reduce an instance of \textsc{3-Bounded-3-SAT} to an instance of \textsc{Perfect Target Set Selection} with thresholds equal to two in polynomial time. In this reduction, the number of vertices and edges of a resulting graph remain linear over the length of the initial formula. In other words, an instance of \textsc{3-Bounded-3-SAT} with $\O(n)$ variables and $\O(m)$ clauses can be reduced to an instance of \textsc{PTSS} with $\O(n+m)$ vertices of constant maximum degree and thresholds equal to two, in polynomial time. This implies that such instances of \textsc{PTSS} cannot be solved in $2^{o(n+m)}$ time.
\end{proof} 
\subsection{Parameterization by \texorpdfstring{$\ell$}{l}}

We now look at \textsc{Target Set Selection} from parameterized point of view. In \cite{Bazgan2014}, Bazgan et al.\ proved that 
\textsc{Target Set Selection} $\ell$ is W[1]-hard with respect to parameter $\ell$, when all dual thresholds are equal to $0$. This result also follows from the proof of W[1]-hardness of \textsc{Cutting $\ell$ Vertices} given by Marx in \cite{Marx2006}, with a somewhat different construction. Inspired by his proof, we show that this result holds even when all thresholds are constant.

\begin{theorem}
\textsc{Target Set Selection} parameterized by $\ell$ is W[1]-hard even when all thresholds are equal to two.
\end{theorem}
\begin{proof}

Let $(G, k)$ be an instance of the \textsc{Clique} problem. In order to provide the reduction, we construct a graph $G'$ in which each vertex corresponds to a vertex or an edge of graph $G$ i.e.\ $V(G')=V(G) \sqcup E(G)$. We add edges in $G'$ between vertices corresponding to $v \in V(G)$ and $e \in E(G)$ if and only if $v$ and $e$ are incident in $G$.

We will refer to the vertex in $G'$ corresponding to an edge $e \in E(G)$ as $v_e \in V(G')$. If a vertex from $G'$ corresponds to a vertex $u\in G$ we refer to it as $v_u$. Slightly abusing notation we will refer 
to the set of vertices in $G'$ corresponding to the vertices $V(G)$ as $V$ and to the set of vertices corresponding to the edges $E(G)$ as $E$, $V\sqcup E=V(G')$.  Consider now an instance  of \textsc{Target Set Selection} for $G'$, with the same $k$, $\ell=k+\binom{k}{2}$ and all thresholds equal to $t=2$.

If $G$ has a clique of size $k$, then selecting corresponding vertices as a target set of $G'$ leads to activation of the vertices corresponding to the edges of the clique. Hence, $k+\binom{k}{2}$ vertices will be activated in total.

Let us now prove that if $G'$ has a target set of size at most $k$ activating at least $\ell=k+\binom{k}{2}$ vertices, then $G$ has a clique on $k$ vertices. Let $S$ be such  target set of $G'$. Denote by $k_v=|S \cap V|$ the number of vertices in $S$ corresponding to the vertices of $G$ and by $k_e=|S \cap E|$ the number of vertices in $S$ corresponding to the edges of $G$, $k_v+k_e\le k$.

Now, we show how to convert any target set $S$ of size at most $k$ activating at least $k+\binom{k}{2}$ vertices into a target set $S'$ such that $|S'|\leq k$, $S' \subseteq V$ and $S'$ activates at least $k+\binom{k}{2}$ vertices. 

Observe that if there is an edge $u_1u_2=e\in E(G)$ such that $v_e \in S$ and $v_{u_1} \in S$ then $S'=S\setminus \{v_e\} \cup \{v_{u_2}\}$  also activates at least $k+\binom{k}{2}$ vertices and the size of $S'$ is at most $k$. Thus we can assume that if $v_{u_1u_2} \in S$, then $v_{u_1},v_{u_2} \not\in S$.   

Observe that any initially not activated vertex in $E$ becomes activated only if all two of its neighbours are activated. It means that any such vertex does not influence the activation process in future. Hence, since $G'$ is bipartite, the activation process always finishes within two rounds, and no vertex in $V$ becomes activated in the second round.

Let $V_1$ be the set of vertices of $V$ that become activated by $S$ in the first round, i.e.\ $V_1=\mathcal{S}_1(S) \setminus \mathcal{S}_0(S) \cap V$. Note that these vertices are activated directly by $k_e$ vertices in $S\cap E$. Let $S_{E,i}$ be the set of vertices in $S\cap E$ that have exactly $i$ endpoints in $V_1$. Denote by $k_{e,i}$ the size of $S_{E,i}$. Then we have $k_{e,0}+k_{e,1}+k_{e,2}=k_e$. Note that if there is a vertex in $S\cap E$ with no endpoints in $V_1$ then one can replace it with any neighbour and size of $S$ will not change and it will activate at least the same number of vertices in $G'$. Thus we can assume that $k_{e,0}=0$.

We show that $|V_1| \le \frac{k_{e,1}}{2}+k_{e,2}$. Indeed, in order to be activated, any vertex from $V_1$ requires at least two vertices from $E$ to be in the target set.  Each vertex from $S_{E,i}$ contributes to exactly $i$ vertices from $V_1$, and the total number of contributions is $k_{e,1}+2k_{e,2}$. This number  should be at least $2|V_1|$. Hence, $|V_1| \le \frac{k_{e,1}}{2}+k_{e,2}$.

Consider $S' = S \setminus E\cup V_1$ i.e.\ we replace all $k_e$ vertices from $E$ with all vertices from $V_1$. Note that $|S'|\leq |S| - \frac{k_{e,1}}{2}$. Vertices from $S_{E,2}$ become activated in the first round since all of them have two endpoints in $S'$. Thus $S'$ is now a target set of size not greater than $k-\frac{k_{e,1}}{2}$ activating at least $\ell-k_{e,1}$ vertices in $G'$. 

Note that any vertex from $S_{E,1}$ can be activated by adding one more vertex to $S'$. Consider set $H=N(S_{E,1})\setminus V_1$. If $|H|\leq \frac{k_{e,1}}{2}$ then consider  $S_1=H\cup S'$. $S_1$ compared to $S'$ will additionally activate all vertices in $S_{E,1}$. Note that $S_1$ is a target set $S$ of size at most $k$ activating at least $\ell$ vertices.

If $|H| > \frac{k_{e,1}}{2}$ then construct $S_1$ from $S'$ by simply adding $\frac{k_{e,1}}{2}$ arbitrary vertices from $H$. Each of these vertices will additionally activate at least one vertex corresponding to edge, thus $S_1$ is a target set of size at most $k$ activating at least $\ell$ vertices. 

We have shown how to transform any target set $S$ activating at least $k+\binom{k}{2}$ vertices in $G'$ into a target set $S_1$ such that $S_1\subseteq V$ and $S_1$ activates at least the same number of vertices in $G'$. As we have shown earlier, no vertex in $E \setminus S_1$ influence the activation process after becoming activated. Then, since $S_1 \cap E = \emptyset$, $S_1$ activates only vertices in $E$ in the first round and the process finishes. Hence, if the instance for $G'$ has a solution, then $G$ has a clique of size $k$.
\end{proof} 
\bibliography{ref-pretty}

\end{document}